%% file: clustering-lb+cap.tex
\documentclass[a4paper,USenglish]{lipics-v2016}
 
\usepackage{microtype}

\usepackage{amsmath,amssymb,amsthm}

\usepackage[utf8]{inputenc}
\usepackage[T1]{fontenc}

\usepackage{booktabs}

\DeclareMathOperator{\opt}{\mathsf{opt}}
\newcommand{\kk}{\texorpdfstring{$k$}{k}}

\Copyright{Clemens Rösner and Melanie Schmidt}

\EventEditors{Ioannis Chatzigiannakis, Christos Kaklamanis, Daniel Marx, and Don Sannella}
\EventNoEds{4}
\EventLongTitle{45th International Colloquium on Automata, Languages, and Programming (ICALP 2018)}
\EventShortTitle{ICALP 2018}
\EventAcronym{ICALP}
\EventYear{2018}
\EventDate{July 9--13, 2018}
\EventLocation{Prague, Czech Republic}
\EventLogo{eatcs}
\SeriesVolume{80}
\ArticleNo{}

\begin{document}

\title{Privacy preserving clustering with constraints}
\author[1]{Clemens Rösner}
\author[1]{Melanie Schmidt}
\affil[1]{Department of Theoretical Computer Science, University of Bonn, Germany\\
\texttt{roesner@cs.uni-bonn.de}, \texttt{melanieschmidt@uni-bonn.de}}
\date{\large Draft, November 2017\vspace*{-0.5cm}}

\maketitle 

\begin{abstract}
The $k$-center problem is a classical combinatorial optimization problem which asks to find $k$ centers such that the maximum distance of any input point in a set $P$ to its assigned center is minimized. The problem allows for elegant $2$-approximations. However, the situation becomes significantly more difficult when constraints are added to the problem. 
We raise the question whether general methods can be derived to turn an approximation algorithm for a clustering problem with some constraints into an approximation algorithm that respects one constraint more. 
Our constraint of choice is privacy: Here, we are asked to only open a center when at least $\ell$ clients will be assigned to it. We show how to combine privacy with several other constraints.

\end{abstract}

\section{Introduction}
Clustering is a fundamental unsupervised learning task: Given a set of objects, partition them into clusters, such that objects in the same cluster are well matched, while different clusters have something that clearly differentiates them. The three classical clustering objectives studied in combinatorial optimization are \emph{$k$-center}, \emph{$k$-median} and \emph{facility location}. Given a point set $P$, $k$-center and $k$-median ask for a set of $k$ centers and an assignment of the points in $P$ to the selected centers that minimize an objective. For $k$-center, the objective is the maximum distance of any point to its assigned center. For $k$-median, it is the sum of the distances of all points to their assigned center (this is called connection cost). Facility location does not restrict the number of centers. Instead, every center (here called facility) has an opening cost. The goal is to find a set of centers such that the connection cost plus the opening cost of all chosen facilities is minimized.
In the unconstrained versions each point will be assigned to its closest center. With the addition of constraints a different assignment is often necessary in order to satisfy the constraints.

A lot of research has been devoted to developing approximation algorithms for these three. The earliest success story is that of $k$-center: Gonzalez~\cite{G85} as well as Hochbaum and Shmoys~\cite{HS86} gave a $2$-approximation algorithm for the problem, while Hsu and Nemhauser~\cite{HN79} showed that finding a better approximation is NP-hard.

Since then, much effort has been made to approximate the other two objectives. Typically, facility location will be first, and transferring new techniques to $k$-median poses additional challenges.
Significant techniques developed during the cause of many decades are LP rounding techniques~\cite{CGTS02,STA97}, greedy and primal dual methods~\cite{JMS02,JV01}, local search algorithms~\cite{AGKMMP04,KPR00}, and, more recently, the use of pseudo-approximation~\cite{LS16}. The currently best approximation ratio for facility location is 1.488~\cite{L13}, while the best lower bound is 1.463~\cite{GK99}. For $k$-median, the currently best approximation algorithm achieves a ratio of 2.675+$\epsilon$~\cite{BPRST17}, while the best lower bound is~$1+\frac{2}{e}\approx 1.736$~\cite{JMS02}.

While the basic approximability of the objectives is well studied, a lot less is known once constraints are added to the picture. Constraints come naturally with many applications of clustering, and since machine learning and unsupervised learning methods become more and more popular, there is an increasing interest in this research topic. It is one of the troubles with approximation algorithms that they are often less easy to adapt to a different scenario than some easy heuristic for the problem, which was easier to understand and implement in the first place. Indeed, it turns out that adding constraints to clustering often requires fundamentally different techniques for the design of approximation algorithms and is a very new challenge altogether. 

A good example for this is the \emph{capacity} constraint: Each center $c$ is now equipped with a capacity $u(c)$, and can only serve $u(c)$ points. This natural constraint is notoriously difficult to cope with; indeed, the standard LP formulations for the problems have an unbounded integrality gap. Local search provides a way out for facility location, leading to $3$- and $5$-approximations for uniform~\cite{ALBGGGJ13} and non-uniform capacities~\cite{BGG12}, and preprocessing together with involved rounding proved sufficient for $k$-center to obtain a $9$-approximation~\cite{CHK12,ABCGMS15}. However, the choice of techniques that turned out to work for capacitated clustering problems is still very limited, and indeed \emph{no} constant factor approximation is known to date for $k$-median.

And all the while, new constraints for clustering problems are proposed and studied. In \emph{private} clustering~\cite{APFTKKZ10}, we demand a lower bound on the number of points assigned to a center to ensure a certain anonymity. The more general form where each cluster has an individual lower bound is called clustering \emph{with lower bounds}~\cite{AS16}. \emph{Fair} clustering~\cite{CKLV17} assumes that points have a protected feature (like gender), modeled by a color, and that we want clusters to be fair in the sense that the ratios between points of different colors is the same for every cluster. Clustering \emph{with outliers}~\cite{CKMN01} assumes that our data contains measurement errors and searches for a solution where a prespecified number of points may be excluded from the cost computation. Other constraints include fault tolerance~\cite{KPS00}, matroid or knapsack constraints~\cite{CLLW16}, must-link and cannot-link constraints~\cite{WCRS01}, diversity~\cite{LYZ10} and chromatic clustering constraints~\cite{DX11,DX15}.

The abundance of constraints and the difficulty to adjust methods for all of them individually asks for ways to \emph{add} a constraint to an approximation algorithm in an oblivious way. Instead of adjusting and reproving known algorithms, we would much rather like to take an algorithm as a black box and ensure that the solution satisfies one more constraint in addition. 
This is a challenging request. We start the investigation of such add-on algorithms by studying private clustering in more detail. 
Indeed, we develop a method to add the privacy constraint to approximation algorithms for constraint $k$-center problems. That means that we use an approximation algorithm as a subroutine and ensure that the final solution will additionally respect a given lower bound. The method has to be adjusted depending on the constraint, but it is oblivious to the underlying approximation algorithm used for that constraint. 

This works for the basic $k$-center problem (giving an algorithm for the private $k$-center problem), but we also show how to use the method when the underlying approximation algorithm is for $k$-center with outliers, fair $k$-center, capacitated $k$-center and fair capacitated $k$-center. We also demonstrate that our method suffices to approximate \emph{strongly private} $k$-center, where we assume a protected feature like in fair clustering, but instead of fairness, now demand that a minimum number of points of each color is assigned to each open center to ensure anonymity for each class individually. 

\subparagraph*{Our Technique}

The general structure of the algorithm is based on standard thresholding~\cite{HS86}, i.e., the algorithm tests all possible thresholds and chooses the smallest for which it finds a feasible solution. For each threshold, it starts with the underlying algorithm and computes a non private solution. Then it builds a suitable network to shift points to satisfy the lower bounds. The approximation ratio of the method depends on the underlying algorithm and on the structure of this network.

The shifting does not necessarily work right away. If it does not produce a feasible solution, then using the max flow min cut theorem, we obtain a set of points for which we can show that the clustering uses too many clusters (and can thus not satisfy the lower bounds). The algorithm then recomputes the solution in this part. Depending on the objective function, we have to overcome different hurdles to ensure that the recomputation works in the sense that it a) makes sufficient progress towards finding a feasible solution and b) does not increase the approximation factor. The process is then iterated until we find a feasible solution. 

\subparagraph*{Results}

We obtain the following results for multiple combinations of privacy with other constraints. 
Note that our definition of $k$-center (see Section~\ref{sec:prelim}) distinguishes between the set of points $P$ and the set of possible center locations $L$. This general case is also called the \emph{$k$-supplier problem}, while classical $k$-center often assumes that $P=L$. Our reductions can handle the general case; whether the resulting algorithm is then for $k$-center or $k$-supplier thus depends on the evoked underlying algorithm. 

\begin{itemize}
\item We obtain a $4$-approximation for private $k$-center with outliers ($5$ for the supplier version). This matches the best known bounds \cite{APFTKKZ10} (\cite{AS16} for the supplier version (this also holds for non-uniform lower bounds)).
\item We compute an $11$-approximation for private capacitated $k$-center (i.e., centers have a lower bound \emph{and} an upper bound), and a $8$-approximation for private uniform capacitated $k$-center (where the upper bounds are uniform, as well). The best known bounds for these two problems are $9$ and $6$~\cite{DHHL17}.
For the supplier version we obtain a $13$-approximation which matches the best known bound~\cite{DHHL17} (for uniform upper bounds a $9$-approximation-algorithm is known~\cite{DHHL17}).
\item We achieve constant factor approximations for private fair capacitated/uncapacitated $k$-center/$k$-supplier clustering. The approximation factor depends on the \emph{balance} of the input point set and the type of upper bounds, it ranges between $10$ in the uncapacitated case where for each color $c$ the number of points with color $c$ is an integer multiple of the number of points with the rarest color and $325$ in the general supplier version with non-uniform upper bounds. To the best of our knowledge, all these combinations have not been studied before.
\item Along the way, we propose constant factor algorithms for general cases of fair clustering. While~\cite{CKLV17} introduces a pretty general model of fairness, it only derives approximation algorithms for inputs with two colors and a balance of $1/t$ for an integer $t$. We achieve ratios of $14$ and $15$ for the general fair $k$-center and supplier problem, respectively. 
\item Finally, we propose the \emph{strongly private $k$-center problem}. As in the fair clustering problem, the input here has a protected feature like gender, modeled by colors. Now instead of a fair clustering, we aim for anonymity for each color, meaning that we have a lower bound for each color. Each open center needs to be assigned this minimum number of points for each color. To the best of our knowledge, this problem has not been studied before; we obtain a $4$-approximation as well as a $5$-approximation for the supplier version.
\end{itemize}

Since our method does not require knowledge of the underlying approximation algorithm, the approximation guarantees improve if better approximation algorithms for the underlying problems are found.
There is also hope that our method could be used for new, not yet studied constraints, with not too much adjustment. 

\subparagraph*{Related Work}

\begin{table}\centering
\begin{tabular}{lllllll}
& \multirow{2}[3]{*}{Vanilla} & \multicolumn{2}{c}{Capacities} & \multirow{2}[3]{*}{Outlier} & \multicolumn{2}{c}{Fair Subset Partition}\\
\cmidrule(lr){3-4}\cmidrule(lr){6-7}
&&  uniform & non-uniform & & $\frac{r}{b}\in\mathbb{N}$&general\\
\midrule
$k$-center & 2~\cite{HS86} & 6~\cite{KS00}& 9~\cite{ABCGMS15} & 2~\cite{CGK16} &\multirow{2}[3]{*}{2~\cite{CKLV17}} & \multirow{2}[3]{*}{12~(Thm.~\ref{thm:7_approx_fair_subset})}\\[0.1cm]
$k$-supplier & 3~\cite{HS86} &  & 11~\cite{ABCGMS15} & 3~\cite{CKMN01}\\[0.5cm]
\end{tabular}
\caption{An overview on the approximation results that we combine with privacy.\label{table:resultsweuse}}
\end{table}

The unconstrained $k$-center problem can be $2$-approximated~\cite{G85,HS86}, and it is NP-hard to approximate it better~\cite{HN79}. The $k$-supplier problem can be $3$-approximated~\cite{HS86}, and this is also tight. 

Capacitated $k$-center was first approximated with uniform upper bounds~\cite{BKP93,KS00}.
Two decades after the first algorithms for the uniform case,~\cite{CHK12} provided the first constant factor approximation for non-uniform capacities. The algorithm was improved and also applied to the $k$-supplier problem in~\cite{ABCGMS15}. In contrast to upper bounds (capacities), lower bounds are less studied. The private $k$-center problem is introduced and $2$-approximated in~\cite{APFTKKZ10}, and non-uniform lower bounds are studied in~\cite{AS16}. 
The $k$-center/$k$-supplier problem with outliers is $3$-approximated in~\cite{CKMN01} alongside approximations to other robust variants of the $k$-center problem. The approximation factor for the $k$-center problem with outliers was improved to $2$ in~\cite{CGK16}.

The fair $k$-center problem was introduced in~\cite{CKLV17}. The paper describes how to approximate the problem by using an approximation for a subproblem that we call fair subset partition problem. Algorithms for this subproblem are derived for two special cases where the number of colors is two, and the points are either perfectly balanced or the number of points of one color is an integer multiple of the number of points of the other color. 

These are the constraints for which we make use of known results. We state the best known bounds and their references in Table~\ref{table:resultsweuse}. Approximation algorithms are also e.g. known for fault tolerant $k$-center~\cite{KPS00} and $k$-center with matroid or knapsack constraints~\cite{CLLW16}.

Relatively little is known about the combination of constraints. Cygan and Kociumaka~\cite{CK14} give a 25-approximation for the capacitated $k$-center problem with outliers. Aggarwal et. al~\cite{APFTKKZ10} give a $4$-approximation for the private $k$-center problem with outliers. Ahmadian and Swamy~\cite{AS16} consider the combination of $k$-supplier with outliers with (non-uniform) lower bounds and derive a $5$-approximation. The paper also studies the $k$-supplier problem with outliers (without lower bounds), and the min-sum-of-radii problem with lower bounds and outliers. Their algorithms are based on the Lagrangian multiplier preserving primal dual method due to Jain and Vazirani~\cite{JV01}.

Ding et. al~\cite{DHHL17} study the combination of capacities and lower bounds as well as capacities, lower bounds and outliers by generalizing the LP algorithms from~\cite{ABCGMS15} and ~\cite{CK14} to handle lower bounds. They give results for several variations, including a $6$-approximation for private capacitated $k$-center and a $9$-approximation for private capacitated $k$-supplier. 

Friggstad, Rezapour, Salavatipour~\cite{FRS16} consider the combination of uniform capacities and non-uniform lower bounds for \emph{facility location} and obtain bicriteria approximations.

\subparagraph*{Outline}

In Section~\ref{sec:prelim}, we introduce necessary notation. Section~\ref{sec:outlier_privacy} then presents our method, applied to the private $k$-center problem with outliers. We choose the outlier version since it is non-trivial but still intuitive and does thus give a good impression on the application of our method. In Section~\ref{sec:combinations}, we then adjust the method to approximate private and fair $k$-center, private and capacitated $k$-center, and $k$-center with all three constraints. In Section~\ref{sec:multiple_lower_bounds}, we consider the strongly private $k$-center problem. We conclude the paper with Section~\ref{sec:outlook} by some remarks on private facility location.

\section{Preliminaries}\label{sec:prelim}

Let $(X,d)$ be a finite metric space, i.e., $X$ is a finite set and $d : X \times X \to \mathbb{R}_{\ge 0}$ is a metric. We use $d(x,T) = \min_{y \in T} d(x,y)$ for the smallest distance between $x \in X$ and a set $T \subseteq X$. For two sets $S, T \subseteq X$, we use $d(S,T) = \min_{x \in S, y \in T} d(x,y)$ for the smallest distance between any pair $x \in S, y \in T$.

Let $P \subseteq X$ be a subset of $X$ called \emph{points} and let $L \subseteq X$ be a subset of $X$ called \emph{locations}. 
An instance of a private \emph{assignment constrained} $k$-center problem  consists of $P$, $L$, an integer $k\in \mathbb{N}$, a lower bound $\ell\in\mathbb{N}$ and possibly more parameters. Given the input, the problem is to compute a set of centers $C \subseteq L$ with $|C| \leq k$ and an \emph{assignment} $\phi: P \rightarrow C$ of the points to the selected centers that satisfies $\ell \le |\phi^{-1}(c)|$ for every selected center $c \in C$, and some specific \emph{assignment restriction}. The solution $C, \phi$ shall be chosen such that
\[
\max_{x \in P} d(x,\phi(x))
\]
is minimized. Different assignment restrictions lead to different constrained private $k$-center problems. The \emph{capacity} assignment restriction comes with an upper bound function $u : L \to \mathbb{N}$ for which we require $\ell \le u(x)$ for all $x \in L$, and then demands $|\phi^{-1}(c)|\le u(c)$. When we have $u(x)=u$ for all $x \in L$ and some $u\in \mathbb{N}$, then we say that the capacities are \emph{uniform}, otherwise, we say they are \emph{non-uniform}.
The \emph{fairness} assignment restriction provides a mapping $\chi : P \rightarrow Col$ of points to colors and then requires that each cluster has the same ratio between the numbers of points with different colors (see Section~\ref{sec:fair} for specifics). 
The \emph{strongly private $k$-center problem} can also be cast as a $k$-center problem with an assignment restriction. Again, the input now additionally contains a mapping $\chi$ of points to colors. Now the assignment is restricted to ensure that it satisfies the lower bound for the points of each color. We even consider the slight generalization where each color has its own lower bound, and call this problem the strongly private $k$-center problem. 

An instance of the \emph{private $k$-center problem with outliers} consists of $P$, $L$, an integer $k\in \mathbb{N}$, a lower bound $\ell$, and a parameter $o$ for the maximum number of outliers. 
The problem is to compute a set of centers $C \subseteq L$ with $|C| \leq k$, a set of outliers $O$ with $|O| \le o$, and an assignment $\phi: P\backslash O \rightarrow C$ of the points that are not outliers to the centers in $C$. The choice of $C, O, \phi$ shall minimize
\[
\max_{x \in P\backslash O} d(x,\phi(x)).
\]

\input{input_privat_outlier}

\input{input_other_constraints}

\input{input_lower_and_upper_clc}

\input{input_fair_clc}

\input{input_fair_and_upper_bounds}

\input{input_multiple_lower_bounds_clc}

\input{input_outlook}

\bibliography{references}
\bibliographystyle{plainurl}

\appendix

\input{input-fl}

\end{document}

%% file: input_privat_outlier.tex
\section{Private \kk-center with Outliers}
\label{sec:outlier_privacy}


\begin{theorem}\label{thm:privat_outlier_kcenter_general}
Assume that there exists an approximation algorithm $A$ for the $k$-center problem with outliers with approximation factor $\alpha$. 

Then for instances $P$, $L$, $k$, $\ell$, $o$ of the private $k$-center problem with outliers, we can compute an $(\alpha + 2)$-approximation in polynomial time.
\end{theorem}

\begin{proof}
Below, we describe an algorithm that uses a threshold graph with threshold $\tau$. We show that for any given $\tau \in \mathbb{R}$, the algorithm has polynomial runtime and, if $\tau$ is equal to $\opt$, the value of the optimal solution, computes an $(\alpha + 2)$-approximation.
Since we know that the value of every solution is equal to the distance between a point and a location, we test all $O(|P||L|)$ possible distances for $\tau$ and return the best feasible clustering returned by any of them.
The main proof is the proof of Lemma~\ref{lemma:privat_outlier_Algo_for_fixed_tau} below, which concludes this proof.
\end{proof}

We now describe the procedure for a fixed value  of $\tau > 0$.

\begin{lemma}
\label{lemma:privat_outlier_Algo_for_fixed_tau}
Assume that there exists an approximation algorithm $A$ for the $k$-center problem with outliers with approximation factor $\alpha$.  
Let $P$, $L$, $k$, $\ell$, $o$ be an instance of the private $k$-center problem with outliers, let $\tau > 0$ and let $\opt$ denote the maximum radius in the optimal feasible clustering for $P$, $L$, $k$, $\ell$, $o$.
We can in polynomial time compute a feasible clustering with a maximum radius of at most $(\alpha + 2) \tau$ or determine $\tau < \opt$.
\end{lemma}

\begin{proof}
The algorithm first uses A to compute a solution without the lower bound: Let $\mathcal{C} =(C,\phi)$ be an $\alpha$-approximate solution for the $k$-center problem with outliers on $P$, $L$, $k$, $o$.
Notice that it can happen that $\mathcal{C}$ contains clusters with fewer than $\ell$ points.

Let $k'=|C|$ (notice that $k'< k$ is possible), $C = \{c_1,\ldots,c_{k'}\}$, and let $C_1,\ldots,C_{k'}$ be the clusters that $\mathcal{C}$ induces, i.e., $C_j := \phi_1^{-1}(c_j)$. Finally, let $r = \max_{x \in P} d(x,\phi(x))$ be the largest distance of any point to its assigned center.
Observe that an optimal solution to the $k$-center problem with outliers can only have a lower objective value than the optimal solution to our problem because we only dropped a condition. Therefore, $\tau \ge \opt$ implies that $r \le \alpha \cdot \opt \le \alpha \cdot \tau$. If we have $r > \alpha \cdot \tau$, we return $\tau < \opt$.

We use $\mathcal{C}$ and $\tau$ to create a threshold graph which we use to either reassign points between the clusters to obtain a feasible solution or to find a set of points $P'$ for which we can show that every feasible clustering with maximum radius $\tau$ uses less clusters than our current solution to cover it. In the latter case we compute another $\alpha$-approximate solution which uses fewer clusters on $P'$ and repeat the process. 
Note that for $\tau < \opt$ such a clustering does not necessarily exist, but for $\tau \ge \opt$ the optimal clustering provides a solution for $P'$ with fewer clusters. If we do not find such a clustering with maximum radius at most $\alpha \cdot \tau$, we return $\tau < \opt$.

We show that every iteration of the process reduces the number of clusters or the number of outliers, therefore the process stops after at most $k \cdot o$ iterations. It may happen that our final solution contains much less clusters than the optimal solution (but it will be an approximate solution for the optimal solution with $k$ centers).

We will use a network flow computation to move points from clusters with more than $\ell$ points to clusters with less than $\ell$ points. Moving a point to another cluster can increase the radius of the cluster. We only want to move points between clusters such that the radius does not increase by too much. More precisely, we only allow a point $p$ to be moved to another cluster $C_i$ if the distance $d(p,C_i)$ between the point and the clusters is at most $2\tau$. This is ensured by the structure of the network described in the next paragraph. Unless stated otherwise, when we refer to distances between a point and a cluster in the following, we mean the distance between the point and the cluster in its original state before any points have been reassigned.

Given $\mathcal{C}$ and $\tau$, we create the threshold graph $G_\tau = (V_\tau,E_\tau)$ as follows. $V_\tau$ consists of a source $s$, a sink $t$, a node $v_i$ for each cluster $C_i$, a node $v_{out}$ for the set of outliers and a node $w_p$ for each point $p \in P$. For all $i \in [k']$, we connect $s$ to $v_i$ if the cluster $C_i$ contains more than $\ell$ points and set the capacity of $(s,v_i)$ to $|C_i|- \ell$. If the cluster $C_i$ contains fewer than $\ell$ points, we connect $v_i$ with $t$ and set the capacity of $(v_i,t)$ to $\ell - |C_i|$. 
Furthermore, we connect $v_i$ with $w_p$ for all $p \in C_i$ and set the capacity of $(v_i,w_p)$ to $1$. 
We also connect $s$ to $v_{out}$ with capacity $o$ and $v_{out}$ with $w_p$ for all $p \in \phi^{-1}(out)$ with capacity $1$.
Whenever a point $p$ and a cluster $C_i$ with $p \notin C_i$ satisfy $d(p,C_i) \le 2 \tau$ (i.e., there is a point $q \in C_i$ that satisfies $d(p,q) \leq 2\tau$), we connect $w_p$ with $v_i$ with capacity $1$. 

Formally the graph $G_\tau = (V_\tau,E_\tau)$ is defined by 
\begin{align} 
V_\tau = &\{v_{out}\} \cup \{v_i \mid 1\leq i \leq k'\} \cup \{w_p \mid p \in P\} \cup \{s,t\}\text{ and}\\ 
E_\tau = &\{(v_i,w_p) \mid p \in C_i\} \cup \{(w_p,v_i) \mid p \notin C_i \wedge d(p,C_i) \le 2\tau\}\\
\cup &\{(v_{out},w_p) \mid \phi(p) = out\}\\
\cup &\{(s,v_{out})\} \cup \{(s,v_i) \mid |C_i|-\ell > 0\} \cup \{(v_i,t) \mid |C_i|-\ell < 0\}.
\end{align}
We define the capacity function $cap: E_\tau \rightarrow \mathbb{R}$ by 
\begin{equation}
    cap(e)=
    \begin{cases}
      \ell - |C_i|, & \text{if}\ e = (v_i,t) \\
      |C_i| - \ell, & \text{if}\ e = (s,v_i) \\
			o, & \text{if}\ e = (s,v_{out})\\
			1 & \text{otherwise.}
    \end{cases}
  \end{equation} 

We use $G = (V,E)$ to refer to $G_\tau$ as $\tau$ is clear from context.
We now compute an integral maximum $s$-$t$-flow $f$ on $G$.
According to $f$ we can reassign points different clusters.

\begin{lemma}
\label{lemma:privat_outlier_reassignment_feasible}
Let $f$ be an integral maximal $s$-$t$-flow on $G$. It is possible to reassign $p$ to $C_i$ for all edges $(w_p,v_i)$ with $f((w_p,v_i)) = 1$. 

The resulting solution has a maximum radius of at most $r + 2 \tau$. If $f$ saturates all edges of the form $(v_i,t)$, then the solution is feasible.
\end{lemma}

\begin{proof}
Let $p \in C_i$.
The choice of capacity $1$ on $(v_i,w_p)$ and flow conservation ensure 
\[
\sum_{(w_p,v_j) \in E} f((w_p,v_j)) \le 1 
\] for $p$. Therefore no point would have to be reassigned to more than one cluster.
Note that for every point $p \in C_i$ that would be reassigned we must have $f((v_i,w_p)) = 1$ and for every edge $(v_i,w_p)$ with $f((v_i,w_p)) = 1$ the point $p$ would be reassigned.

For any $1 \le j \le k'$, let $p \in C_i$ be any point which we want to reassign to $C_j$. Then we must have $(w_p,v_j) \in E$ and therefore there must be a point $q \in C_j$ with $d(p,q) \le 2\tau$.
Thus we have 
\[d(p,c_j) \le d(p,q) + d(q,c_j) \le 2\tau + r = r + 2\tau.\]

Now assume that $f$ saturates all edges of the form $(v_i,t)$ and let $1 \le i \le k'$. 
If $E$ contains the edge $(v_i,t)$, then it can not contain the edge $(s,v_i)$ and therefore all incoming edges of $v_i$ are of the form $(w_p,v_i)$. Flow conservation then implies that the number of points reassigned to $C_i$ minus the points reassigned away from $C_i$ is equal to $f((v_i,t))$, which increases the number of points in $C_i$ to $\ell$.

If $E$ contains the edge $(s,v_i)$, then it can not contain the edge $(v_i,t)$ and therefore all outgoing edges of $v_i$ are of the form $(v_i,w_p)$. Flow conservation then implies that the number of points reassigned away from $C_i$ minus the points reassigned to $C_i$ is equal to $f((s,v_i))$, which reduces the number of points in $C_i$ to at least $\ell$.

If $E$ contains neither $(s,v_i)$ nor $(v_i,t)$, then the number of points in $C_i$ is equal to $\ell$ and does not change (the points may change, but their number does not).

In all three cases $C_i$ contains at least $\ell$ points after the reassignment.
\end{proof}

If $f$ saturates all edges of the form $(v_i,t)$ in $G$, then we reassign points according to Lemma~\ref{lemma:privat_outlier_reassignment_feasible} and return the new clustering.

Otherwise we look at the residual network $G_{f}$ of $f$ on $G$. 
Let $V'$ be the set of nodes in $G_{f}$ which can not be reached from $s$. We say cluster $C_i$ \emph{belongs} to $V'$ if $v_i \in V'$, and a point $p \in C_i$ is \emph{adjacent} to $V'$ if $w_p \in V'$ and $v_i \notin V'$. Let $C(V')$ denote the set of clusters belonging to $V'$. Let $k'' = |C(V')|$. We say a point $p$ belongs to $V'$ if the cluster $C_i$ with $p \in C_i$ belongs to $V'$. Let $P(V')$ and $P_A(V')$ denote the set of points that belong to $V'$ and the set of points adjacent to $V'$.

\begin{lemma}
\label{lemma:privat_outlier_properties_V'}
Any clustering on $P$ with maximum radius at most $\tau$ that contains at least $\ell$ points in every cluster uses fewer than $k''$ clusters to cover all points in $P(V')$.
\end{lemma}

\begin{proof}
We first observe that $V'$ must have the following properties:
\begin{itemize}
\item $v_i \in V'$ and $(w_p,v_i)\in E$ implies $w_p \in V'$.
\item $w_p \in V'$, $(w_p,v_i) \in E$ and $f((w_p,v_i)) > 0$ implies $v_i \in V'$.
\item $w_p \in V'$ for some $p \in C_i$ and $v_i \notin V'$ implies $f((v_i,w_p)) = 1$.
\end{itemize}
The first property follows from the fact that $f$ can only saturate $(w_p,v_i)$ if $f$ also saturates $(v_j,w_p)$ for $p \in C_j$. So, either $(w_p,v_i)$ is not saturated, which means that $v_i$ can be reached from any vertex that reaches $w_p$, or $(w_p,v_i)$ \emph{is} saturated, which means that the only incoming edge of $w_p$ in $G_{f}$ is $(v_i,w_p)$. In both cases, if $v_i \in V'$, then $w_p \in V'$.
The second property follows since $f((w_p,v_i)) > 0$ implies $(v_i,w_p) \in E(G_{f})$.
The third property is true since we defined $cap((v_i,w_p)) = 1$.

This implies that a reassignment due to Lemma~\ref{lemma:privat_outlier_reassignment_feasible} would reassign all points adjacent to $V'$ to clusters in $C(V')$ and moreover all reassignments from points in $P(V') \cup P_A(V')$ would be to clusters in $C(V')$. Let $n_i$ denote the number of points that would be assigned to $C_i$ after the reassignment.
Then $|P(V')| + |P_A(V')| = \sum_{C_i \in C(V')} n_i$.

Now we argue that this sum is smaller than $k'' \cdot \ell$ by observing that each $n_i \le \ell$ and at least one $n_i$ is strictly smaller than $\ell$. 

Let $C_i$ be a cluster with more than $\ell$ points after the reassignment. Then $(s,v_i)$ is not saturated by $f$ and $v_i$ can be reached from $s$ in $G_{f}$. Therefore after the reassignment no cluster $C_i \in C(V')$ would contain more than $\ell$ points; in other words, $n_i > \ell$ implies $C_i \notin C(V')$.

Let $C_i$ be a cluster which would still contain fewer than $\ell$ points after the reassignment. This implies that $f$ does not saturate the edge $(v_i,t)$. Therefore $t$ can be reached from $v_i$ and since $f$ is a maximum $s$-$t$ flow, $v_i$ can not be reached from $s$. We must have $v_i \in V'$.

Because we assumed that the reassignment does not satisfy all lower bounds, at least one such cluster has to exist. 
This implies
\[|P(V')| + |P_A(V')| = \sum_{C_i \in C(V')} n_i < k'' \cdot \ell.\]
Which means that the clusters in $C(V')$ and $P_A(V')$ do not contain enough points to satisfy the lower bound in $k''$ clusters.

By definition of $G$ and $V'$, for two points $p,q$ with $p \in P(V')$ and $d(p,q) \le 2\tau$ we must have $q \in P(V') \cup P_A(V')$. 
Let $\mathcal{C}'$ be a clustering that abides the lower bounds and has a maximal radius of at most $\tau$. Then every cluster $C'$ in $\mathcal{C}'$ that contains at least one point from $P(V')$ can only contain points from $P(V') \cup P_A(V')$. 
Therefore $\mathcal{C}'$ must contain fewer than $k''$ clusters which contain at least one point from $P(V')$.
\end{proof}

If we have $\tau \ge \opt$, then Lemma~\ref{lemma:privat_outlier_properties_V'} implies that the optimal solution covers all points in $P(V')$ with fewer than $k''$ clusters. An $\alpha$-approximative solution on the point set $P(V')$ with at most $k''-1$ clusters which contains at most $o$ outliers is then $\alpha$-approximative for $P(V')$.

Unfortunately, we do not know how many outliers an optimal clustering has in $P(V')$.
We therefore involve the outliers $\phi^{-1}(out)$ in our new computation as well.
Let $o' = |\phi^{-1}(out)|$ denote the current number of outliers.
We obtain the following Lemma through a counting argument.

\begin{lemma}
\label{lemma:outlier_properties_V'2}
We call a cluster special if it contains at least one point from $P(V')$ or only contains points from $\phi^{-1}(out)$.
Let $\mathcal{C}'$ be a clustering on $P$ with a maximum radius of at most $\tau$ on all special clusters that respects the lower bounds, has at most $o$ outliers and consists of at most $k$ clusters out of which at most $k''$ are special. 
If $\mathcal{C}'$ has exactly $k''$ special clusters, then $\mathcal{C}'$ has at most $o' - 1$ outliers in $P(V) \cup \phi^{-1}(out)$.
\end{lemma}
\begin{proof}
Assume the clustering contains exactly special $k''$ clusters. Each of these clusters has to contain at least $\ell$ points from $P(V') \cup P_A(V') \cup \phi^{-1}(out)$. We know 
\[|P(V') \cup P_A(V') \cup \phi^{-1}(out)| \le |P(V') \cup P_A(V')| + o' < k'' \ell + o'.\]
So there remain at most $o'-1$ unclustered points in $P(V) \cup \phi^{-1}(out)$.
\end{proof}

Now we need to show that such a clustering exists if $\tau \ge opt$ is the case.

\begin{lemma}
If $\tau \ge opt$, then there exists a clustering $\mathcal{C}'$ on $P$ with a maximum radius at most $\tau$ on all special clusters that respects the lower bounds, has at most $o$ outliers and consists of at most $k$ clusters out of which at most $k''$ are special.
\end{lemma}

\begin{proof}
We look at an optimal clustering $\mathcal{C}_{opt}$. The only way $\mathcal{C}_{opt}$ can violate a condition is if it contains $k'''> k''$ special clusters.
Lemma~\ref{lemma:privat_outlier_properties_V'} implies that $\mathcal{C}_{opt}$ contains at least $k'''-k''$ clusters that contain only points in $\phi^{-1}(out)$.
If all clusters in $\mathcal{C}_{opt}$ are special we know $P = P_A(V') \cup P(V') \cup \phi^{-1}(out)$.
We arbitrarily select $k'''-k''$ clusters from $\mathcal{C}_{opt}$ that contain only points in $\phi^{-1}(out)$, declaring all points in them as outliers and closing the corresponding centers. This leaves us with $k''$ clusters which contain at least $k'' \cdot \ell$ points. Since $P = P_A(V') \cup P(V') \cup \phi^{-1}(out)$ this leaves at most $o' - 1$ outliers.
Otherwise, if $\mathcal{C}_{opt}$ contains at least one cluster $C$ which is not special, we add all outliers from $P \setminus (P_A(V') \cup P(V') \cup \phi^{-1}(out)$ to $C$.
Again we arbitrarily select $k'''-k''$ clusters from $\mathcal{C}_{opt}$ that contain only points in $\phi^{-1}(out)$, declaring all points in them as outliers and closing the corresponding centers.
By creation there are no unclustered points in $P \setminus (P_A(V') \cup P(V') \cup \phi^{-1}(out)$ and exactly $k''$ special clusters with radius at most $\tau$. Therefore this clustering contains at most $o'-1$ outliers and has at most $k$ clusters.
\end{proof}

We now use $A$ again to compute new solutions without the lower bound: Let $\mathcal{C}'_1 =(C'_1,\phi'_1)$ be an $\alpha$-approximate solution for the $k$-center problem with outliers on $P(V')\cup \phi^{-1}(out)$, $L$, $k''-1$, $o$ and let  $\mathcal{C}'_2 =(C'_2,\phi'_2)$ be an $\alpha$-approximate solution for the $k$-center problem with outliers on $P(V')\cup \phi^{-1}(out)$, $L$, $k''$, $o'-1$. Let $r'_i = \max_{x \in P(V')\cup \phi^{-1}(out) } d(x,\phi'_i(x))$.

 Note that in case $\tau < \opt$, it can happen that no such clustering exists or that we obtain $r'_i > \alpha \cdot \tau$ for both $i=1$ and $i=2$. We then return $\tau < \opt$. 
Otherwise for at least one $i \in \{1,2\}$ $\mathcal{C}'_i$ must exist together with $r'_i \le \alpha \cdot \tau$.

If $\mathcal{C}'_2$ exists and we have $r'_2 \le \alpha \cdot \tau$ we replace $C(V')$ by $C'_2$ in $\mathcal{C}$ and adjust $\phi$ accordingly to obtain $\mathcal{C}_1 = (C_1,\phi_1)$ with $C_1 = (C\setminus C(V')) \cup C'_2$ and 
\begin{equation}
    \phi_1(p)=
    \begin{cases}
      \phi'_2(p) & \text{if}\ p \in P(V')\cup \phi^{-1}(out)\\
      \phi(p) & \text{otherwise.} \\
    \end{cases}
  \end{equation}
Otherwise, if $\mathcal{C}'_1$ exists, we have $r'_1 \le \alpha \cdot \tau$ and either $\mathcal{C}'_2$ does not exist or we have $r'_2 > \alpha \cdot \tau$, we analogous replace $C(V')$ by $C'_1$ to obtain $\mathcal{C}_1$.

\begin{lemma}
If we did not return $\tau < \opt$, then $\mathcal{C}_1$ is a solution for the $k$-center problem with outliers on $P$, $L$, $k$, $o$ and we have $r_1 = \max_{x \in P} d(x,\phi_1(x)) \le \alpha \cdot \tau$.
\end{lemma}
\begin{proof}
$\mathcal{C}$ is a solution for the $k$-center problem with outlier on $P$, $L$, $k$, $o$ with $r < \alpha \cdot \tau$ and since we did not return $\tau < \opt$, we must have $r'_i \le \alpha \tau$ for the chosen $i \in \{1,2\}$. 
\end{proof}

We iterate the previous process with the new clustering $\mathcal{C}_1$ until we either determine $\tau < \opt$ or the reassignment of points according to Lemma~\ref{lemma:privat_outlier_reassignment_feasible} yields a feasible solution. Since each iteration reduces the number of clusters or keeps the same number of clusters and reduces the number of outliers, the process terminates after at most $k\cdot o$ iterations.
\end{proof}

\begin{corollary}
We can compute a $4$-approximation for instances of the private $k$-center problem with outliers and a $5$-approximation for instances of the private $k$-supplier problem in polynomial time. 
\end{corollary}

\begin{proof}
Follows from Theorem~\ref{thm:privat_outlier_kcenter_general} together with the $2$-approximation for $k$-center with outliers in~\cite{CGK16} and the $3$-approximation for $k$-supplier with outliers in~\cite{CKMN01}.
\end{proof}

%% file: input_other_constraints.tex
\section{Combining Privacy with other Constraints}
\label{sec:combinations}
We want to take the general idea from Section~\ref{sec:outlier_privacy} and instead of outliers we want to combine privacy with other restrictions on the clusters. 
Given a specific restriction $\mathcal{R}$ and an approximation algorithm $A$ for the $k$-center problem with restriction $\mathcal{R}$ with approximation factor $\alpha$ we ask: Can we similar to Section~\ref{sec:outlier_privacy} combine $A$ with the use of a threshold graph to compute an $O(\alpha)$-approximation for the private $k$-center problem with restriction $\mathcal{R}$? 


In Section~\ref{sec:outlier_privacy} we made use of two properties of a clustering with outliers. In Lemma~\ref{lemma:privat_outlier_reassignment_feasible} we used that reassigning points to another cluster never increases the number of outliers and in Lemma~\ref{lemma:privat_outlier_properties_V'} we used that outliers have the somewhat local property that computing a new clustering on the points $V'$ from a subset of the clusters together with the set of outliers can not create more outliers on the remaining points.

In this section we now take a look at restriction properties which are similarly local, and show how to combine them with privacy.

%% file: input_lower_and_upper_clc.tex
\subsection{Privacy and Capacities}\label{sec:lower_upper_bound}

\begin{theorem}
\label{thm:kcenter_general}
Assume that there exists an approximation algorithm $A$ for the capacitated $k$-center problem with approximation factor $\alpha$. 
Then we can compute an $(\alpha + 2)$-approximation for the private capacitated $k$-center problem in polynomial time.
\end{theorem}

\input{input_lower_and_upper_content}

%% file: input_lower_and_upper_content.tex

\begin{proof}
Let $P$, $L$, $k$, $u$, $\ell$ be an instance of the private capacitated $k$-center problem.

Analogous to Section~\ref{sec:outlier_privacy} we use a threshold graph with threshold $\tau$ and show that for any given $\tau \in \mathbb{R}$ the algorithm has polynomial runtime and, if $\tau$ is equal to $\opt$, the value of the optimal solution, computes an $(\alpha + 2)$-approximation.
Since we know that the value of the optimal solution is equal to the distance between a point and a location, we test all $O(|P||L|)$ possible distances for $\tau$ and return the best feasible clustering returned by any of them.
The main proof is the proof of Lemma~\ref{lemma:Algo_for_fixed_tau} below. The lemma then concludes the proof.
\end{proof}

We now describe the procedure for a fixed value  of $\tau > 0$.

\begin{lemma}
\label{lemma:Algo_for_fixed_tau}
Assume that there exists an approximation algorithm $A$ for the capacitated $k$-center problem with approximation factor $\alpha$. 

Let  $P$, $L$, $k$, $u$, $\ell$ be an instance of the private capacitated $k$ center problem and let $\tau > 0$. and let $\opt$ denote the maximum radius in the optimal feasible clustering for $P$, $L$, $k$, $u$, $\ell$.
We can in polynomial time compute a feasible clustering with a maximum radius of at most $(\alpha+ 2) \tau$ or determine $\tau < \opt$.
\end{lemma}

\begin{proof}
The algorithm first uses A to compute a solution without the lower bound: Let $\mathcal{C} =(C,\phi)$ be an $\alpha$-approximate solution for the capacitated $k$-center problem on $P$, $L$, $k$, $u$.


Again let $k'=|C|$, $C = \{c_1,\ldots,c_{k'}\}$, let $C_1,\ldots,C_{k'}$ be the clusters that $\mathcal{C}$ induces, i.e., $C_j := \phi_1^{-1}(c_j)$ and let $r = \max_{x \in P} d(x,\phi(x))$ be the largest distance of any point to its assigned center.
If we have $r > \alpha \cdot \tau$, we return $\tau < \opt$.

Given $\mathcal{C}$ and $\tau$, we create, similar to Section~\ref{sec:outlier_privacy}, a threshold graph $G_{\tau} = (V_{\tau},E_{\tau})$ by 
\begin{align} 
V_\tau = &\{v_i \mid 1\leq i \leq k'\} \cup \{w_p \mid p \in P\} \cup \{s,t\}\text{ and}\\ 
E_\tau = &\{(v_i,w_p) \mid p \in C_i\} \cup \{(w_p,v_i) \mid p \notin C_i \wedge d(p,C_i) \le 2\tau\}\\
\cup &\{(s,v_i) \mid |C_i|-\ell > 0\} \cup \{(v_i,t) \mid |C_i|-\ell < 0\}.
\end{align}
We define the capacity function $cap: E_\tau \rightarrow \mathbb{R}$ by 
\begin{equation}
    cap(e)=
    \begin{cases}
      \ell - |C_i|, & \text{if}\ e = (v_i,t) \\
      |C_i| - \ell, & \text{if}\ e = (s,v_i) \\
			1 & \text{otherwise.}
    \end{cases}
  \end{equation}
The only difference to Section~\ref{sec:outlier_privacy} is that we do not have any outliers.
We use $G = (V,E)$ to refer to $G_\tau$ as $\tau$ is clear from context.
We now compute an integral maximum $s$-$t$-flow $f$ on $G$.
According to $f$ we can reassign points to different clusters.

Analogous to Lemma~\ref{lemma:privat_outlier_reassignment_feasible} we obtain the following lemma.
\begin{lemma}
\label{lemma:reassignment_feasible}
Let $f$ be an integral maximal $s$-$t$-flow on $G$. It is possible to reassign $p$ to $C_j$ for all edges $(w_p,v_j)$ with $f((w_p,v_j)) = 1$ . 

The resulting solution has a maximum radius of at most $r + 2 \tau$. If $f$ saturates all edges of the form $(v_i,t)$, then the solution is feasible.
\end{lemma}

In case $f$ saturates all edges of the form $(v_i,t)$ we reassign points according to Lemma~\ref{lemma:reassignment_feasible} and return the new clustering.

Otherwise, we look at the residual network $G_f$ of $f$ on $G$.
We define $V'$ and $k''$ as before, i.e., $V'$ is the set of nodes in $G_{f_i}$ which can not be reached from $s$, and $k''$ is the number of clusters which belong to $V'$. As before, we obtain the following lemma.

\begin{lemma}
\label{lemma:properties_V'}
Any clustering on $P$ with maximum radius at most $\tau$ that respects the lower bounds 
uses fewer than $k''$ clusters to cover all points in $P(V')$.
\end{lemma}

In case we have $\tau \ge \opt$ this implies that the optimal solution covers all points in $P(V')$ with fewer than $k''$ clusters. An $\alpha$-approximative solution on the point set $P(V')$ with at most $k''-1$ clusters which abides only the upper bounds is then $\alpha$-approximative for $P(V')$.

We now use $A$ again to compute a new solution without the lower bound: Let $\mathcal{C}' =(C',\phi')$ be an $\alpha$-approximate solution for the capacitated $k$-center problem on $P(V')$, $L$, $k''-1$, $u$. Let $r' = \max_{x \in P(V')} d(x,\phi'(x))$.
 Note that in case $\tau < \opt$, it can happen that no such clustering exists or that we obtain $r' > \alpha \cdot \tau$. We then return $\tau < \opt$.

Otherwise we replace replace $C(V')$ by $C'$ in $\mathcal{C}$ and adjust $\phi$ accordingly to obtain $\mathcal{C}_1 = (C_1,\phi_1)$ with $C_1 = (C\setminus C(V')) \cup C'$ and 
\begin{equation}
    \phi_1(p)=
    \begin{cases}
      \phi'(p) & \text{if}\ p \in P(V')\\
      \phi(p) & \text{otherwise.} \\
    \end{cases}
  \end{equation}

\begin{lemma}
In case we did not return $\tau < \opt$, $\mathcal{C}_1$ is a solution for the capacitated $k$-center problem on $P$, $L$, $k$, $u$ and we have $r_1 = \max_{x \in P} d(x,\phi_1(x)) \le \alpha \cdot \opt$.
\end{lemma}

We iterate the previous process with new clustering $\mathcal{C}_1$ until we either determine $\tau < \opt$ or the reassignment of points according to Lemma~\ref{lemma:reassignment_feasible} yields a feasible solution. Since the number of clusters is reduced in each iteration, the process terminates after at most $k$ iterations.
\end{proof}


\begin{corollary}\label{cor:capa}
We can compute an $11$-approximation for instances of the private capacitated $k$-center problem in polynomial time. 

If the upper bounds are uniform, too, then we can compute an $8$-approximation.
\end{corollary}
\begin{proof}
Follows from Theorem~\ref{thm:kcenter_general} together with the $9$-approximation for capacitated $k$-center in~\cite{ABCGMS15}. For uniform upper bounds, capacitated $k$-center can be $6$-approximated~\cite{KS00}, leading to a guarantee of $8$.
\end{proof}

\begin{corollary}\label{cor:capb}
We can compute a $13$-approximation for instances of the private capacitated $k$-supplier problem in polynomial time. 

\end{corollary}
\begin{proof}
Follows from Theorem~\ref{thm:kcenter_general} together with the $11$-approximation for capacitated $k$-center in~\cite{ABCGMS15}.
\end{proof}

%% file: input_fair_clc.tex
\input{input_fair_content_one}
\begin{theorem}
\label{thm:7_approx_fair_subset}
A $12$-approximation for the fair subset partition problem can be computed in polynomial time.
If $b_c = 1$ for at least one color $c \in Col$, then a $2$-approximation for the fair subset partition problem can be computed in polynomial time (even if $|Col|>2$).
\end{theorem}
\input{input_fair_content_two}

\begin{theorem}
\label{thm:fair2}
Assume that there exists an approximation algorithm $A$ for the fair subset partition problem  with approximation factor $\alpha$. 
Then we can compute a $(3\alpha + 4)/(3\alpha + 5)$-approximation for the private fair $k$-center/supplier problem in polynomial time.
\end{theorem}

\begin{corollary}
\label{cor:private_fair}
We can compute a $40$-approximation for instances of the private and fair $k$-center problem and a $41$-approximation for instances of the private and fair $k$-supplier problem in polynomial time.

If $b_c = 1$ for at least one color $c \in Col$, the approximation factors improve to $10$ and $11$.
\end{corollary}

\begin{proof}
All results follow from Theorem~\ref{thm:fair2} together with an approximation for the fair subset partition problem.

In case $b_c = 1$ for some $c \in Col$ we use the $2$-approximation from Theorem~\ref{thm:7_approx_fair_subset} and for the general case use the $12$-approximation for the fair subset partition problem from Theorem~\ref{thm:7_approx_fair_subset}.
\end{proof}

%% file: input_fair_content_one.tex
\subsection{Privacy and Fairness}
\label{sec:fair}
Fair clustering was introduced in~\cite{CKLV17}. The idea is that there are one or more protected features of the objects, and that the composition of all clusters should be fair with respect to the protected features. Formally, the protected features are modeled by colors. \cite{CKLV17} defines fair clustering problems for the case of two colors, i.e., two protected features.

We consider the general version with an arbitrary amount of colors.
Thus in the fair version of the $k$-center problem, in addition to $P$, $L$ and $k$, each point in $P$ is colored. We denote the set of colors by $Col$ and let $\chi : P \rightarrow Col$ assign the points to their colors. For a subset $P' \subseteq P$ and a color $c \in Col$, let $c(P') = \{p \in P' \mid \chi(p) = c\}$.
A clustering $\mathcal{C}$ is considered fair if the ratios between points with different colors is the same in every cluster, i.e., for every pair $c,d \in Col$ and every $C \in \mathcal{C}$, we have $\frac{|c(C)|}{|d(C)|} = \frac{|c(P)|}{|d(P)|}$.


Again we adjust our method in order to apply it to the fair $k$-center problem to obtain the following lemma.

\begin{lemma}
\label{lemma:fair}
Assume that there exists an approximation algorithm $A$ for the fair $k$-center problem with approximation factor $\alpha$. 
Then for instances $P$, $L$, $k$, $Col$, $\chi$, $\ell$  of the private and fair $k$-center problem, we can compute a $(3\alpha + 2)$-approximation in polynomial time.
\end{lemma}

\begin{proof}
Analogous to Section~\ref{sec:outlier_privacy} we use a threshold graph with threshold $\tau$ and show that for any given $\tau \in \mathbb{R}$, the algorithm has polynomial runtime, and, if $\tau$ is equal to $\opt$, the value of the optimal solution, computes an $(3\alpha + 2)$-approximation.
Since we know that the value of the optimal solution is equal to the distance between a point and a location, we test all $O(|P||L|)$ possible distances for $\tau$ and return the best feasible clustering returned by any of them.
The main proof is the proof of Lemma~\ref{lemma:fair_fixed_tau} below. The lemma then concludes the proof.
\end{proof}

We now describe the procedure for a fixed value  of $\tau > 0$.

\begin{lemma}
\label{lemma:fair_fixed_tau}
Assume that there exists an approximation algorithm $A$ for the fair $k$-center problem with approximation factor $\alpha$. 

Let  $P$, $L$, $k$, $Col$, $\chi$, $\ell$ be an instance of the private and fair $k$-center problem, let $\tau > 0$ and let $\opt$ denote the maximum radius in the optimal feasible clustering for $P$, $L$, $k$, $Col$, $\chi$, $\ell$.
We can in polynomial time compute a feasible clustering with a maximum radius of at most $3\alpha \cdot \tau + 2 \tau$ or determine $\tau < \opt$.
\end{lemma}

\begin{proof}
The algorithm first uses A to compute a solution without the lower bound: Let $\mathcal{C} =(C,\phi)$ be an $\alpha$-approximate solution for the fair $k$-center problem on $P$, $L$, $k$, $Col$, $\chi$.

Again let $k'=|C|$, $C = \{c_1,\ldots,c_{k'}\}$, let $C_1,\ldots,C_{k'}$ be the clusters that $\mathcal{C}$ induces, i.e., $C_j := \phi_1^{-1}(c_j)$ and let $r = \max_{x \in P} d(x,\phi(x))$ be the largest distance of any point to its assigned center.

If we have $r > \alpha \cdot \tau$, we return $\tau < \opt$.

Reassigning a point to a different cluster can result in both the old and the new cluster not being fair anymore. Therefore we unfortunately can not simply create a threshold graph and move points from one cluster to another.

For every $c \in Col$ let $b_c = \frac{|c(P)|}{gcd(\{|d(P)| \mid d\in Col\})}$, then it is easy to see that in every feasible clustering every cluster contains a multiple of $b:=\sum_{c\in Col} b_c$ points.

Instead of moving single points between clusters we want to move sets which contain $b_c$ points with color $c$ for every $c \in Col$, thus keeping the clustering fair.

\begin{definition}
A subset $P' \subseteq P$ is called a fair subset of $P$, if for every $c \in Col$ $P'$ contains exactly $b_c$ points with color $c$, i.e., for all $c \in Col$ we have $|P' \cap c(P)| = b_c$.
\end{definition}

We use $\mathcal{C}$ to arbitrarily partition $P$ into fair sets such that all points in the same set belong to the same cluster in $\mathcal{C}$.
Let $F = \{F_1,\ldots\}$ denote these sets. 
By construction the distance between any two points in the same set is at most $2\alpha \tau$.

Given $\mathcal{C}$, $S$ and $\tau$, we create the threshold graph $G_\tau = (V_\tau,E_\tau)$ similar to Section~\ref{sec:outlier_privacy} by 
\begin{align} 
V_\tau = &\{v_{out}\} \cup \{v_i \mid 1\leq i \leq k'\} \cup \{f_i \mid F_i \in F\} \cup \{s,t\}\text{ and}\\ 
E_\tau = &\{(v_i,f_j) \mid F_j \subseteq C_i\} \cup \{(f_j,v_i) \mid F_j \cap C_i = \emptyset \wedge d(C_i,F_j) \le 2\tau\}\\
\cup &\{(s,v_i) \mid |C_i|-\ell > 0\} \cup \{(v_i,t) \mid |C_i|-\ell < 0\}.
\end{align}
We define the capacity function $cap: E_\tau \rightarrow \mathbb{R}$ by 
\begin{equation}
    cap(e)=
    \begin{cases}
      \left\lceil \frac{\ell - |C_i|}{b}\right\rceil, & \text{if}\ e = (v_i,t) \\
      \left\lfloor \frac{|C_i| - \ell}{b}\right\rfloor, & \text{if}\ e = (s,v_i) \\
			\\
			1 & \text{otherwise.}
    \end{cases}
  \end{equation}
The difference to the threshold graph in Section~\ref{sec:outlier_privacy} is that we do not have outliers and include the nodes $f_i$ for the fair sets instead of nodes for the points. We also changed the capacities, such that the capacities of edges of the form $(v_i,t)$ now represent how many additional fair sets $C_i$ needs to satisfy the lower bound, while capacities of edges of the form $(s,v_i)$ now represent how many fair sets $C_i$ can give away and still contain at least $\ell$ points.

We use $G = (V,E)$ to refer to $G_\tau$ as $\tau$ is clear from context.
We now compute an integral maximum $s$-$t$-flow $f$ on $G$.
According to $f$ we can reassign fair subsets to different clusters.

Analogous to Lemma~\ref{lemma:privat_outlier_reassignment_feasible} we obtain the following lemma.
\begin{lemma}
\label{lemma:fair_reassignment_feasible}
Let $f$ be an integral maximal $s$-$t$-flow on $G$. It is possible to reassign $F_i$ to $C_j$ for all edges $(f_i,v_j)$ with $f((f_i,v_j)) = 1$. 

The resulting solution has a maximum radius of at most $3r + 2 \tau$. If $f$ saturates all edges of the form $(v_i,t)$, then the solution is feasible.
\end{lemma}
Note that in contrast to Lemma~\ref{lemma:privat_outlier_reassignment_feasible} we obtained a new radius of at most $3r + 2 \tau$ because when we add a fair subset $F_j$ to a cluster $C_i$ the maximum distance of a point $p$ in $F_j$ to $c_i$ is at most $\max_{q \in F_j} d(p,q) + d(F_j,C_i) + r \le 2r + 2 \tau + r$.

In case $f$ saturates all edges of the form $(v_i,t)$ we reassign points according to Lemma~\ref{lemma:fair_reassignment_feasible} and return the new clustering.

Otherwise, we again look at the residual network $G_f$ of $f$ on $G$. We define $V'$ and $k''$ as before, i.e., $V'$ is the set of nodes in $G_f$ which can not be reached from $s$, and $k''$ is the number of clusters which belong to $V'$. As before, we obtain the following lemma. 

\begin{lemma}
\label{lemma:fair_properties_V'}
Any fair clustering on $P$ with maximum radius at most $\tau$ that respects the lower bounds 
uses fewer than $k''$ clusters to cover all points in $P(V')$.
\end{lemma}

In case we have $\tau \ge \opt$ this implies that the optimal solution covers all points in $P(V')$ with fewer than $k''$ clusters. 

A fair $\alpha$-approximative solution on the point set $P(V')$ with at most $k''-1$ clusters is then $\alpha$-approximative for $P(V')$.

We now use $A$ again to compute a new solution without the lower bound: Let $\mathcal{C}'_1 =(C'_1,\phi'_1)$ be an $\alpha$-approximate solution for the fair $k$-center problem $P(V')$, $L$, $k''-1$, $Col$, $\chi$. Let $r' = \max_{x \in P(V')} d(x,\phi'(x))$.

 Note that in case $\tau < \opt$, it can happen that no such clustering exists or that we obtain $r' > \alpha \cdot \tau$. We then return $\tau < \opt$.
Otherwise we replace replace $C(V')$ by $C'$ in $\mathcal{C}$ and adjust $\phi$ accordingly to obtain $\mathcal{C}_1 = (C_1,\phi_1)$ with $C_1 = (C\setminus C(V')) \cup C'$ and 
\begin{equation}
    \phi_1(p)=
    \begin{cases}
      \phi'(p) & \text{if}\ p \in P(V')\\
      \phi(p) & \text{otherwise.} \\
    \end{cases}
  \end{equation}

\begin{lemma}
In case we did not return $\tau < \opt$, $\mathcal{C}_1$ is a solution for the fair $k$-center problem on $P$, $L$, $k$, $Col$, $\chi$ and we have $r_1 = \max_{x \in P} d(x,\phi_1(x)) \le \alpha \cdot \opt$.
\end{lemma}

We iterate the previous process with the new clustering $\mathcal{C}_1$ until we either determine $\tau < \opt$ or the reassignment of points according to Lemma~\ref{lemma:fair_reassignment_feasible} yields a feasible solution. Since each iteration reduces the number of clusters, the process terminates after at most $k$ iterations.
\end{proof}

\subparagraph{The Fair Subset Partition Problem\\}
\label{sec:fair_subset}
Let the fair subset partition problem denote the problem which given a set of points $P$, a set of colors $Col$ and a function $\chi: P \rightarrow Col$ computes a partition $P= \bigcup_{i \in \{1,\ldots,\frac{n}{b}\}} F_i$ into fair subsets together with a representative center $y_i \in P$ for each fair subset $F_i$ and minimizes $\max\{d(y_i,p) \mid i \in \{1,\ldots,\frac{n}{b}\} \wedge p \in F_i\}$. 

The following Lemma is a generalization to results in~\cite{CKLV17}. In the case that $Col$ contains $2$ colors $red$ and $blue$ with $b_{red} = 1$ or $b_{blue} = 1$ they show a $2$-approximation for the fair subset partition problem.

%% file: input_fair_content_two.tex
\begin{proof}
Let $c \in Col$ be an arbitrary color.
We use an algorithm by Khuller and Sussmann~\cite{KS00} to compute a $5$-approximation to the capacitated $k$-center problem, with $c(P)$ as the set of points and locations, together with $k = \frac{|c(P)|}{b_c}$ and a soft uniform upper bound of $b_c$. This enforces that every cluster contains exactly $b_c$ points, each of which has color $c$. Let $\{(C_i,c_i) \mid 1 \le i \le k\}$ denote these sets together with their computed center.

For each color $d \in Col\setminus\{c\}$ we now compute a matching to add $b_d$ points with color $d$ to each of these sets $C_i$.
Our matching instances consist of complete bipartite graphs $G_d = (C \cup D, E = \{\{u,v\}\mid u \in C \wedge v \in D\})$. $C$ consists of $b_d$ vertices for each subsets $C_i$, while $D$ contains a vertex for every point with color $d$. The weight of an edge $\{u,v\}$ between $u \in C$ and $v \in D$ is the distance between the point corresponding to $v$ and the center of the set corresponding to $u$.
We now compute the smallest weight $w$ such that $G_d$ restricted to edges with weight at most $w$ contains a perfect matching. As there are at most $|c(P)||d(P)|$ different weights this can be tested in polynomial time by checking for each weight $w$ if there exists a perfect matching in the graph which contains only the edges with weight at most $w$.

We now take such a perfect matching and according to the matching we add the points with color $d$ to the sets of points with color $c$. By construction this adds $b_d$ points to each of the sets.
It is now left to show that the radius of every created set is at most $12$ times the optimal radius.
Let $P = \bigcup_{i \in \{1,\ldots,\frac{n}{b}\}} F_i$ be the optimal solution to the fair subset partition problem and let $opt$ be its value. Then $\{F_i \cap c(P) \mid i \in \{1,\ldots,\frac{n}{b}\}\}$ is a solution to the capacitated $k$-center problem on $c(P)$ with $k = \frac{|c(P)|}{b_c}$ and a soft uniform upper bound of $b_c$. With the same centers this yields a value of at most $opt$. If we enforce that the centers have to be in $c(P)$ this yields a value of at most $2opt$. The computed $5$-approximation therefore has a radius of at most $10 opt$.

Let $B \subseteq C$. Since we made $b_d$ copies for each of the fair subsets, $B$ has to contain vertices out of at least $\left\lceil \frac{|B|}{b_d}\right\rceil$ such subsets which represent at least $b_c \left\lceil \frac{|B|}{b_d}\right\rceil$ points with color $c$.
Therefore there are at least $\left\lceil \frac{|B|}{b_d}\right\rceil$ fair sets in the optimal solution, which contain at least one of the points represented by $B$.
These fair sets of the optimal solution then contain at least $b_d \left\lceil \frac{|B|}{b_d}\right\rceil$ many points with color $d$.
 Let $p$ be an arbitrary point with color $d$ in one of these fair subsets. Since two points in the same optimal fair subsets have a distance of at most $2opt$, there exists a point $q$ represented by $B$ with $d(p,q) \le 2opt$. The distance of $p$ to the center corresponding to $q$ is therefore at most $12opt$.

Therefore in the subgraph of $G_d$ which contains only edges with a weight at most $12opt$ the neighborhood of $B$ contains at least $b_d \left\lceil \frac{|B|}{b_d}\right\rceil \ge |B|$ vertices. The marriage theorem~\cite{H35} therefore shows that the subgraph contains a perfect matching.

If $b_c = 1$ for at least one color $c \in Col$, then we choose $c$ in the beginning and therefore cluster $c(P)$ into sets which each contain exactly one points.
Then for any point $p$ the point with color $c$ in the same optimal fair subset will be one of the computed centers.
Therefore in the subgraph of $G_d$ which contains only edges with a weight at most $2opt$ the neighborhood of $B$ contains at least $b_d \left\lceil \frac{|B|}{b_d}\right\rceil \ge |B|$ vertices. The marriage theorem'~\cite{H35} therefore shows that the subgraph contains a perfect matching.
\end{proof}

\subparagraph{Approximating the Fair $k$-Center Problem}\label{sec:approxfair}
We now use the approximation algorithm for the fair subset partitioning problem to compute approximations for fair $k$-center/$k$-supplier.
\begin{corollary}
\label{cor:fair_approx}
We can compute a $14$-approximation for instances of the fair $k$-center problem and a $15$-approximation for instances of the the fair $k$-supplier problem in polynomial time.
In case $b_c = 1$ for at least one color $c \in Col$ the approximation factors improve to $4$ and $5$.
\end{corollary}

\begin{proof}
The proof follows from Theorem~\ref{thm:7_approx_fair_subset} together with the proof from~\cite{CKLV17}.
In~\cite{CKLV17} they showed that given an $\alpha$-approximation to the fair subset partition problem one can compute an $(\alpha + 2)$-approximation to the fair $k$-center problem with a $2$-approximation for the $k$-center problem~\cite{G85}.

The same proof can be used analogous to show that given an $\alpha$-approximation to the fair subset partition problem one can compute an $(\alpha + 3)$-approximation to the fair $k$-supplier problem with a $3$-approximation for the $k$-supplier problem~\cite{HS86}.
\end{proof}

Note that in the proof of Lemma~\ref{lemma:fair_fixed_tau} we assumed that the radius of the fair subsets is the same as the radius of the computed approximation. If we instead use an approximation algorithm for the fair subset partition problem with approximation factor $\alpha$ to compute an $\alpha +2$ approximation for the fair $k$-center problem or an $\alpha + 3$ approximation for the fair $k$-supplier problem we obtain the following lemma.

%% file: input_fair_and_upper_bounds.tex
\subsection{Privacy, Fairness and Capacities}
\label{sec:fair_upper_lower}
In this section we consider instances of the private capacitated and fair $k$-center problem.

We let $P$, $L$, $k$, $u$, $Col$, $\chi$, $\ell$ be an instance of the private capacitated and fair $k$-center problem.

We know from Section~\ref{sec:fair} that in every fair clustering, every cluster contains an integer multiple of $b = \sum_{c\in Col} b_c$ points and can be partitioned into fair subsets.

Since every cluster must contain an integer multiple of $b$ points, we assume without loss of generality that the lower bound $\ell$ as well as all upper bounds $\{u(p) \mid p \in P\}$ are integer multiples of $b$ as well.

We can therefore look at a fair clustering in two layers, where the first layer consists of partitioning $P = \bigcup_{i \in \{1,\ldots,\frac{n}{b+r}\}} F_i$ of $P$ into  fair subsets $F_i$ and the second layer consists of clustering these fair subsets. 

We now show a couple useful properties.

\begin{lemma}
\label{lemma:fair_mapping}
Let $P = \bigcup_{i \in \{1,\ldots,\frac{n}{b}\}} F_i$ and $P = \bigcup_{i \in \{1,\ldots,\frac{n}{b}\}} G_i$ be two partitions of $P$ into fair subsets, then there exists a bijective mapping $\pi: \{1,\ldots,\frac{n}{b}\} \rightarrow \{1,\ldots,\frac{n}{b}\}$ such that for each $i \in \{1,\ldots,\frac{n}{b}\}$ we have $F_i \cap G_{\pi(i)} \neq \emptyset$.
\end{lemma}

\begin{proof}
Let $G = (V \cup W,E)$ be a bipartite graph defined by $V = \{v_1,\ldots,v_\frac{n}{b}\}$, $W = \{w_1,\ldots,w_\frac{n}{b}\}$ and $\{v_i,w_j\} \in E \Leftrightarrow F_i \cap G_j \neq \emptyset$. Then the existence of a perfect matching in $G$ is equivalent to the existence of a mapping $\pi$ as described.
If we set costs $c$ to the edges by $c(\{v_i,w_j\}) = |F_i \cap G_j|$ we can see that for every subset of $V' \subseteq V$ and every subset $W' \subseteq W$ the total cost of all edges adjacent to $V'$ is equal to $|V'| (b+r)$ and the total cost of all edges adjacent to $W'$ is equal to $|W'|(b)$.
Therefore the neighborhood of $V' \subseteq V$ contains at least $|V'|$ nodes.
Hall's "`Marriage Theorem"'~\cite{H35} then concludes the proof.
\end{proof}

\begin{lemma}
\label{lemma:fair_subset_exchange}
Let $\bigcup_{i \in \{1,\ldots,\frac{n}{b}\}} G_i = P$ be a partition of $P$ into fair subsets with a maximum diameter $d$ and assume that there exists a feasible clustering.
Then there exists a clustering $\mathcal{C}$ of $P$ with a radius of at most $opt + d$ which in addition to the lower and upper bounds satisfies that for each $i \in \{1,\ldots,\frac{n}{b}\}$ all points in $G_i$ are part of the same cluster.
\end{lemma}
\begin{proof}
Let $\mathcal{C}_{opt}= \{C,\phi\}$ be the optimal feasible clustering and let $P= \bigcup_{i \in \{1,\ldots,\frac{n}{b}\}} F_i$ be a corresponding partition into fair subsets. 
Lemma~\ref{lemma:fair_mapping} shows that there exists a bijective mapping $\pi: \{1,\ldots,\frac{n}{b}\} \rightarrow \{1,\ldots,\frac{n}{b}\}$ such that for each $i \in \{1,\ldots,\frac{n}{b}\}$ we have $F_i \cap G_{\pi(i)} \neq \emptyset$.
For all $i \in \{1,\ldots,\frac{n}{b}\}$ we replace $F_i$ by $G_{\pi(i)}$ in its cluster of the optimal clustering to create the new clustering $\mathcal{C}= \{C,\phi'\}$.
Formally $\phi'$ is defined as follows. For all $i \in \{1,\ldots,\frac{n}{b}\}$ and $p \in G_{\pi(i)}$ we have $\phi'(p) = \phi(q)$ for some $q \in F_i$. 
Note that $\phi'$ is well defined since for all $i \in \{1,\ldots,\frac{n}{b}\}$ we have $\phi(q) = \phi(q')$ for all $q,q' \in F_i$.

Since replacing a fair subset with a different fair subset does not change the number of points in a cluster $\mathcal{C}$ is a feasible clustering and by construction satisfies that for each $i \in \{1,\ldots,\frac{n}{b}\}$ all points in $G_i$ are part of the same cluster.

We know that for all $i \in \{1,\ldots,\frac{n}{b}\}$ we have $F_i \cap G_{\pi(i)} \neq \emptyset$. Let $q \in F_i \cap G_{\pi(i)}$ then we have $d(q,\phi(q)) \le opt$ and $d(p,q) \le d$ for all $p \in G_{\pi(i)}$. By the triangle inequality we immediately obtain $d(p,\phi'(p) = \phi (q)) \le opt + d$.
\end{proof}

The idea is to compute a partition of $P = \bigcup_{i \in \{1,\ldots,\frac{n}{b}\}} F_i$ into fair subsets with small diameter and then use an approximation algorithm to compute a clustering on $(\{f_i \mid 1 \le i \le \frac{n}{b}\}, L, k, \frac{\ell}{b}, \frac{u}{b})$, where $f_i$ is a new point representing $F_i$.

Through the construction we immediately obtain the following lemma.
\begin{lemma}
\label{lemma:approx_fair_subsets}
Let $P = \bigcup_{i \in \{1,\ldots,\frac{n}{b}\}} F_i$ be a partition of $P$ into fair subsets.
Let $F = \{f_i \mid 1 \le i \le \frac{n}{b}\}$ be the set of centers points corresponding to $\{F_i \mid 1 \le i \le \frac{n}{b}\}$. Let $d(f_i,q) = max_{p\in F_i} d(p,q)$ for all $q \in P$.
Let $\mathcal{C} = \{C,\phi\}$ be a solution for the $k$-center problem on $F, L, k$ with a maximum radius $rad$.
Then $\mathcal{C}' = \{C,\phi'\}$ with $\phi'(p) = \phi(f_i)$ for all $p \in F_i$ and all $i\in \{1,\ldots,\frac{n}{b}\}$ is a solution for the fair $k$-center problem on $P, L, k, Col, \chi$ with a maximum radius $rad$.

Analogous let $\mathcal{C} = \{C,\phi\}$ be a solution for the fair $k$-center problem on $P, L, k, Col, \chi$ with $\phi'(p) = \phi(q)$ for all $p,q \in F_i$ for all $i\in \{1,\ldots,\frac{n}{b}\}$ and a maximum radius $rad$.
Then $\mathcal{C}' = \{C,\phi'\}$ with $\phi'(f_i) = \phi(p)$ for $p \in F_i$ is a solution for the $k$-center problem on $F, L, k$ with a maximum radius $rad$.
\end{lemma}

Lemma~\ref{lemma:approx_fair_subsets} gives us a direct correspondence between clusterings on $F$ and clusterings on $P$ in which for each $i \in \{1,\ldots,\frac{n}{b}\}$ all points in $F_i$ belong to the same cluster. More over a cluster in a clustering on $P$ contains exactly $b$ times as many points as the corresponding cluster in the clustering on $F$.
Since we assumed $\ell$ and all $\{u(p) \mid p \in P\}$ to be integer multiples of $b$, the optimal solution for the private capacitated $k$-center problem on $F, L, k, \frac{u}{b}, \frac{\ell}{b}$ therefore directly corresponds to the best solution for the private capacitated and fair $k$-center problem on $P, L, k, u, Col, \chi, \ell$, where for each $i \in \{1,\ldots,\frac{n}{b}\}$ all points in $F_i$ belong to the same cluster.
An $\alpha$-approximate solution for the private capacitated $k$-center problem on
$F, L, k, \frac{u}{b}, \frac{\ell}{b}$ therefore yields an $\alpha$-approximate solution 
 for the private capacitated and fair $k$-center problem on $P, L, k, u, Col, \chi, \ell$, where for each $i \in \{1,\ldots,\frac{n}{b}\}$ all points in $F_i$ have to belong to the same cluster.

\begin{lemma}
\label{lemma:fair_lower_upper}
Assume that there exists an approximation algorithm $A$ for the private capacitated $k$-center problem with approximation factor $\alpha$.
Assume that there exists an approximation algorithm $B$ for the fair subset partition problem with approximation factor $\beta$.
Then for instances $P, L, k, u, Col, \chi, \ell$ of the private capacitated and fair $k$-center problem, we can compute an $\alpha(2\beta + 1)$-approximation in polynomial time.
\end{lemma}
\begin{proof}
We know that the radius of the optimal solution of the fair subset partition problem is at most the radius of the optimal solution to the fair $k$-center problem.
We use $B$ to compute a partition $P = \bigcup_{i \in \{1,\ldots,\frac{n}{b}\}} F_i$ into fair subsets with a maximum diameter of at most $2 \beta opt$.
Lemma~\ref{lemma:fair_subset_exchange} implies that the best solution to the private capacitated and fair $k$-center problem on $P, L, k, u, Col, \chi, \ell$, where for each $i \in \{1,\ldots,\frac{n}{b}\}$ all points in $F_i$ have to belong to the same cluster, has a maximum radius of at most $(2 \beta + 1) opt$.
We then use $A$ to compute an approximation $\mathcal{C}$ on $F, L, k, \frac{u}{b}, \frac{\ell}{b}$ which has a maximum radius of at most $\alpha (2 \beta + 1) opt$.
The solution $\mathcal{C}'$ corresponding to $\mathcal{C}$ according to Lemma~\ref{lemma:approx_fair_subsets} then must be an $\alpha (2 \beta + 1)$-approximation to the private capacitated and fair $k$-center problem on $P, L, k, u, Col, \chi, \ell$.
\end{proof}

\begin{corollary}
\label{cor:private_cap_fair}
We can compute an $O(1)$-approximation for instances of the private capacitated and fair $k$-center/$k$-supplier in polynomial time. 
\end{corollary}

\begin{proof}
For the $k$-center problem we use the $9$-approximation by~\cite{DHHL17} for the problem with uniform lower bound and non uniform upper bounds and the $6$-approximation by~\cite{DHHL17} for the problem with uniform lower bound and uniform upper bounds.
Together with Lemma~\ref{lemma:fair_lower_upper} and the $12$-approximation for the fair subset partition problem from Theorem~\ref{thm:7_approx_fair_subset} this yields an approximation factor of $9 (2 \cdot 12 + 1) = 225$.
In case of a uniform upper bound this reduces the approximation factor to $150$.

In case $b_c = 1$ for some $c \in Col$ the $2$-approximation for the fair subset partition problem from Theorem~\ref{thm:7_approx_fair_subset} improves the approximation factor to $45$ for non-uniform upper bounds and $30$ for uniform upper bounds.

For the $k$-supplier problem we use the $13$-approximation by~\cite{DHHL17} for the problem with uniform lower bound and non uniform upper bounds and the $9$-approximation by~\cite{DHHL17} for the problem with uniform lower bound and uniform upper bounds and obtain the following approximation factors for the fair $k$-supplier problem with non-uniform upper bounds and uniform lower bounds.

In the general case we obtain approximation factors of $325$ and $225$ and in case $b_c = 1$ for some $c \in Col$ we obtain approximation factors of $65$ and $45$.
\end{proof}

%% file: input_multiple_lower_bounds_clc.tex
\section{Strongly Private \kk-center}
\label{sec:multiple_lower_bounds}
Similar to Section~\ref{sec:fair} and Section~\ref{sec:fair_upper_lower} we assume that instances of the strongly private $k$-center problem contain, in addition to $P$, $L$ and $k$, a set of colors $Col$ and a function $\chi: P \rightarrow Col$ which assigns a color to each of the points.
In order to preserve the privacy although additional information about each point is know we demand that each cluster contains enough representatives of each color.

Formally, the strongly private $k$-center problem consists of an instance of the $k$-center problem together with a set of colors $Col$, a function $\chi : P \rightarrow Col$ and a lower bound $\ell_i$ for each color $i \in Col$, where the problem is to compute a set of centers $C \subseteq L$ with $|C| \le k$ and an assignment $\phi: P \rightarrow C$ of the points to the selected centers that satisfies $\ell_i \le \phi^{-1}(x) \cap \chi^{-1}(i)$ for all $i \in Col$ and all $x \in C$ and minimizes 
\[
\max_{x \in P} d(x,\phi(x)).
\]

We again adjust our method from Section~\ref{sec:outlier_privacy} in order to apply it to the strongly private $k$-center problem and obtain the following lemma.

\begin{theorem}
\label{thm:colorful}
Assume that there exists an approximation algorithm $A$ for the $k$-center problem with approximation factor $\alpha$. 
Then we can compute an $(\alpha + 2)$-approximation for the strongly private $k$-center problem in polynomial time.
\end{theorem}

\input{input_multiple_lower_bounds_content}

%% file: input_multiple_lower_bounds_content.tex
\begin{proof}
Let $P$, $L$, $k$, $Col$, $\chi$, $\{\ell_i\mid i \in Col\}$ be an instance of the strongly private $k$-center problem.

Analogous to Section~\ref{sec:outlier_privacy} we use threshold graphs with threshold $\tau$ and show that for any given $\tau \in \mathbb{R}$, the algorithm has polynomial runtime, and, if $\tau$ is equal to $\opt$, the value of the optimal solution, computes an $(\alpha + 2)$-approximation.
Since we know that the value of the optimal solution is equal to the distance between a point and a location, we test all $O(|P||L|)$ possible distances for $\tau$ and return the best feasible clustering returned by any of them.
The main proof is the proof of Lemma~\ref{lemma:colorful_fixed_tau} below. The lemma then concludes the proof.
\end{proof}

We now describe the procedure for a fixed value  of $\tau > 0$.

\begin{lemma}
\label{lemma:colorful_fixed_tau}
Assume that there exists an approximation algorithm $A$ for the $k$-center problem with approximation factor $\alpha$. 

Let $P$, $L$, $k$, $Col$, $\chi$, $\{\ell_i\mid i \in Col\}$  be an instance of the strongly private $k$-center problem, let $\tau > 0$ and let $\opt$ denote the maximum radius in the optimal feasible clustering for $P$, $L$, $k$, $Col$, $\chi$, $\{\ell_i\mid i \in Col\}$.
We can in polynomial time compute a feasible clustering with a maximum radius of at most $(\alpha + 2) \tau$ or determine $\tau < \opt$.
\end{lemma}

\begin{proof}
The algorithm first uses A to compute a solution without the lower bounds: Let $\mathcal{C} =(C,\phi)$ be an $\alpha$-approximate solution for the $k$-center problem on $P$, $L$, $k$. 
Again let $k'=|C|$, $C = \{c_1,\ldots,c_{k'}\}$, let $C_1,\ldots,C_{k'}$ be the clusters that $\mathcal{C}$ induces, i.e., $C_j := \phi_1^{-1}(c_j)$ and let $r = \max_{x \in P} d(x,\phi(x))$ be the largest distance of any point to its assigned center. 
If we have $r > \alpha \cdot \tau$, we return $\tau < \opt$.
For every color $i \in Col$ and a set $Q \subseteq P$ we denote by $Q^i$ the set of points in $Q$ with color $i$, i.e., $Q^i := Q \cap \chi^{-1}(i)$.
Given $\mathcal{C}$ and $\tau$, we create, similar to Section~\ref{sec:outlier_privacy}, a threshold graph $G_{\tau,i} = (V_{\tau,i},E_{\tau,i})$ for every $i \in Col$ by 
\begin{align} 
V_{\tau,i} = &\{v_j \mid 1\leq j \leq k'\} \cup \{w_p \mid p \in P^i\} \cup \{s,t\}\text{ and}\\ 
E_{\tau,i} = &\{(v_j,w_p) \mid p \in C^i_j\} \cup \{(w_p,v_j) \mid p \in P^i \setminus C_j \wedge d(p,C_j) \le 2\tau\}\\
\cup &\{(s,v_j) \mid |C_j \cap \chi^{-1}(i)|-\ell_i > 0\} \cup \{(v_j,t) \mid |C_j \cap \chi^{-1}(i)|-\ell_i < 0\}.
\end{align}
We define the capacity functions $cap_i: E_{\tau,i} \rightarrow \mathbb{R}$ by 
\begin{equation}
    cap(e)=
    \begin{cases}
      \ell_i - |C_j \cap \chi^{-1}(i)|, & \text{if}\ e = (v_j,t) \\
      |C_j \cap \chi^{-1}(i)| - \ell_i, & \text{if}\ e = (s,v_j) \\
			1 & \text{otherwise.}
    \end{cases}
  \end{equation}
The only difference to Section~\ref{sec:outlier_privacy} is that we do not have outliers and create a separate threshold graph for every color.

We use $G_i = (V_i,E_i)$ to refer to $G_{\tau,i}$ as $\tau$ is clear from context.
We now compute integral maximum $s$-$t$-flows $f_i$ on $G_i$.
According to $f_i$ we can reassign points of color $i$ to different clusters.

Analogous to Lemma~\ref{lemma:privat_outlier_reassignment_feasible} we obtain the following lemma.
\begin{lemma}
\label{lemma:colorful_reassignment_feasible}
Let $f_i$ be an integral maximal $s$-$t$-flow on $G_i$. It is possible to reassign $p$ to $C_j$ for all edges $(w_p,v_j)$ with $f((w_p,v_j)) = 1$. 

The resulting solution has a maximum radius of at most $r + 2 \tau$. If $f_i$ saturates all edges of the form $(v_i,t)$, then the solution contains at least $\ell_i$ points of color $i$ in every cluster.
\end{lemma}

If for all $i \in Col$, $f_i$ saturates all edges of the form $(v_j,t)$ in $G_i$, then we reassign points according to Lemma~\ref{lemma:colorful_reassignment_feasible} and return the new clustering.
Note that for each $i \in Col$ $f_i$ would only suggest to reassigns points of color $i$. Therefore the reassignments according to the flows computed for different colors do not interfere with each other.

Otherwise chose an arbitrary color $i \in Col$ such that $f_i$ does not saturate all edges of the form $(v_j,t)$ in $G_i$.
We again look at the residual network $G_{f_i}$ of $f_i$ on $G_i$. We define $V'$ and $k''$ as before, i.e., $V'$ is the set of nodes in $G_{f_i}$ which can not be reached from $s$, and $k''$ is the number of clusters which belong to $V'$. As before, we obtain the following lemma.

\begin{lemma}
\label{lemma:colorful_properties_V'}
Any clustering on $P$ with maximum radius at most $\tau$ that contains at least $\ell_i$ points of color $i$ in every cluster uses fewer than $k''$ clusters to cover all points in $P(V')$.
\end{lemma}

In case we have $\tau \ge \opt$ this implies that the optimal solution covers all points in $P(V')$ with fewer than $k''$ clusters. 

An $\alpha$-approximative solution on the point set $P(V')$ with at most $k''-1$ clusters is therefore $\alpha$-approximative for $P(V')$.

We now use $A$ again to compute a new solution without the lower bounds: Let $\mathcal{C}' =(C',\phi')$ be an $\alpha$-approximate solution for the $k$-center problem on $P(V')$, $L$, $k''-1$. Let $r' = \max_{x \in P(V')} d(x,\phi'(x))$.

 Note that in case $\tau < \opt$, it can happen that no such clustering exists or that we obtain $r' > \alpha \cdot \tau$. We then return $\tau < \opt$. 
Otherwise we replace $C(V')$ by $C'$ in $\mathcal{C}$ and adjust $\phi$ accordingly to obtain $\mathcal{C}_1 = (C_1,\phi_1)$ with $C_1 = (C\setminus C(V')) \cup C'$ and 
\begin{equation}
    \phi_1(p)=
    \begin{cases}
      \phi'(p) & \text{if}\ p \in P(V')\\
      \phi(p) & \text{otherwise.} \\
    \end{cases}
  \end{equation}

\begin{lemma}
In case we did not return $\tau < \opt$, $\mathcal{C}_1$ is a solution for the $k$-center problem $P$, $L$, $k$ and we have $r_1 = \max_{x \in P} d(x,\phi_1(x)) \le \alpha \cdot \tau$.
\end{lemma}

We iterate the previous process with the new clustering $\mathcal{C}_1$ until we either determine $\tau < \opt$ or the reassignment of points according to Lemma~\ref{lemma:colorful_reassignment_feasible} yields a feasible solution. Since each iteration reduces the number of clusters, the process terminates after at most $k$ iterations.
Note that the color $i \in Col$, according to which we define the set $V'$ as the set of nodes in $G_{f_i}$ which can not be reached from $s$, does not have to be the same for every iteration, instead in each iteration the color can be chosen arbitrarily among all colors $i \in Col$ for which $f_i$ does not saturate all edges of the form $(v_j,t)$ in $G_i$.
\end{proof}

\begin{corollary}\label{cor:strongly}
We can compute a $4$-approximation for instances of the strong private $k$-center problem and a $5$-approximation for instances of the strongly private $k$-supplier problem in polynomial time.
\end{corollary}
\begin{proof}
Follows from Theorem~\ref{thm:colorful} together with the $2$-approximation for $k$-center and the $3$-approximation for $k$-supplier, both in~\cite{G85}.
\end{proof}

%% file: input_outlook.tex
\section{Conclusion and open questions}\label{sec:outlook}
We have studied $k$-center with capacities, fairness and outliers and have coupled these constraints with privacy; in addition, we proposed strongly private $k$-center. An obvious open question is to improve the approximation guarantee of the coupling process; this is in particular interesting when combining more than two constraints as in the private capacitated and fair $k$-center problem. Another straightforward direction would be to study the generalization of privacy to arbitrary lower bounds, where each cluster has its individual lower bound on the number of necessary points to assign to it when opened.
It would also be interesting to study general methods to add other constraints to clustering problems. And of course, extending our methods to other clustering objectives is open, too. Our algorithms rely on the threshold graph; removing it seems difficult at first glance. However, in Appendix~\ref{sec:fl}, we demonstrate how to add privacy to capacitated \emph{facility location}, albeit under a restriction: The method only works if the lower bound $\ell$ and all upper bounds $u(c)$ satisfy $\ell \le u(c)/2$. If this is not true, then it induces a capacity violation (by a factor of at most $2$). We raise the question whether adding privacy to facility location can be done without this condition. The method in Appendix~\ref{sec:fl} does not easily extend to variants with a restricted number of centers; so the next question then would be whether it can be combined with the idea we developed for $k$-center, in order to add privacy for an objective like $k$-median.

%% file: input-fl.tex
\section{(Metric) Facility Location}\label{sec:fl}

Let $P = L$, $f \in \mathbb{N}$, $\ell \in \mathbb{N}$, $u \in \mathbb{N}$ with $\ell \le \frac{1}{2}u$ be an instance of the private capacitated facility location problem with uniform upper and uniform lower bounds and uniform facility opening costs. 

Let $\mathcal{C}=(C,\phi)$ be a $\gamma$-approximation for the private facility location problem on $P,L,f,\ell$. We set $k'=|C|$, name the $k'$ facilities $c_1,\ldots,c_{k'}$ and define the partitioning $C_1,\ldots,C_{k'}$ by $C_i = \phi_i^{-1}(c_i)$.

We define an instance where all points are translated to their centers. So we let $P'$ contain $|C_i|$ copies of $c_i$ for each $i \in [k']$.
More precisely, place a point $p_x$ at location $\phi(x)$ for all $x \in P$ and call the resulting set $P'$.
Note that we use $P'$ in order to simplify the analysis and although it is not fully supported by the definition, where $P$ is a \emph{subset} of $X$ we will use the same terminology, when we talk about clusterings on $P'$.

\begin{lemma}\label{lem:transfersolutions}
Let $\mathcal{C'}=(C',\phi')$ be any clustering for $P'$.
Transfer this clustering to $P$ by setting $C''=C'$ and $\phi''(x) = \phi'(p_x)$ for all $x \in P$. Then
\[
\sum_{x \in P} d(x,\phi''(x)) \le \sum_{x \in P} d(x,\phi(x)) + \sum_{x\in P'} d(x,\phi'(x)).
\]
\end{lemma}
\begin{proof}
Let $x \in P$ be any point. Then
\[
d(x,\phi''(x)) = d(x,\phi'(p_x)) \le d(x,p_x) + d(p_x,\phi'(p_x))
\]
by the triangle inequality. Summing this over all $x \in P$, we get that
\[
\sum_{x \in P} d(x,\phi''(x)) \le \sum_{x \in P} d(x,\phi(x)) + \sum_{x \in P'} d(x,\phi'(x)).
\]
\end{proof}

Assume that we have \emph{soft capacities}, i.e., that we can open a center multiple times. Then we compute a solution to the upper bounded facility location problem on $P',L,f,u$ in the following way. 
Firstly, we open a center at every location in $C$. This costs $f \cdot k'$ opening cost. To this center, we assign $|C_i| \mod u$ points that lie at $c_i$. 
After this step, the number of points that are not assigned yet from $C_i$ is a multiple of $u$ (this is true for all $i\in [k])$. We can thus satisfy all their demand by opening at most $n/u$ additional centers. Since any feasible solution opens at least $n/u$ centers, $n/u \le k_opt$. We thus pay additional $f \cdot k_{opt}$ for opening the missing centers, where $k_{opt}$ denotes the number of centers the optimal solution opens. Again, assigning the points costs nothing. 

\begin{corollary}
There is a solution to the upper bounded facility location problem on $P,L,f,u$ with soft capacities which costs at most 
\[
\sum_{x \in P} d(x,\phi(x)) + f \cdot (k' + k_{opt}).
\]
\end{corollary}
\begin{proof}
Follows from the above discussion and Lemma~\ref{lem:transfersolutions}.
\end{proof}

We want to reconcile this solution with the lower bound solution.
The facilities from the second step are valid because they contain $u \ge \ell$ points. The facilities from the first step might be invalid. We will reassign some of the points to different centers at the same location. This costs nothing.

There are two cases. If there is only one center at a location, then we know that it got at least $\ell$ points because there are at least $\ell$ points at the same location in $P'$. Otherwise, we have at most one center at the location that is not full, and at least one center that is full. We can thus reassign up to $u/2$ points, ensuring that we get two facilities which have at least $u/2\ge \ell$ points; the rest of the facilities at this location will remain full.
Thus, at no additional cost, we get a solution that respects both upper and lower bounds.

\begin{corollary}
There is a solution to the private capacitated facility location problem on $P,L,f,u,\ell$ with $2\ell \le u$ and with soft capacities which costs at most 
\[
\sum_{x \in P} d(x,\phi(x)) + f \cdot (k' + k_{opt}).
\]
\end{corollary} 

Now assume we do not have soft capacities. Then we consider the solution computed by the above soft capacity algorithm; it partitions the points into at most $2k'$ clusters $U_1,\ldots,U_{2k}$. For each $U_i$, we pick the best center in $U_i$ as its center. Say that the points in $U_i$ were previously assigned to center $c \notin U_i$ and are now assigned to $c'$. Furthermore, let $c''$ be the point in $U_i$ that is closest to $c$. Then we have 
\[
\sum_{x \in U_i} d(x,c') \le \sum_{x \in U_i} d(x,c''),
\]
and for each point $x \in U_i$,
\[
d(x,c'') \le d(x,c) + d(c'',c) \le 2 d(x,c)
\]
because $c''$ is the closest point to $c$ in $U_i$. This implies
\[
\sum_{x \in U_i} d(x,c') \le 2 \sum_{x \in U_i} d(x,c),
\]
so the assignment cost goes up by a factor of at most two.

\begin{corollary}
There is a solution to the private capacitated facility location problem on $P,L,f,u,\ell$ with $2\ell \le u$ which costs at most 
\[
2 \sum_{x \in P} d(x,\phi(x)) + f \cdot (k' + k_{opt}) \le 2 \gamma OPT_L + OPT \le (2 \gamma + 1) OPT,
\]
where $OPT_L$ is the cost of an optimal solution for the private facility location problem on $P,L,f,\ell$ and $OPT$ is the cost of an optimal solution for the private capacitated facility location problem on $P,L,f,u,\ell$.
\end{corollary}